\documentclass[a4paper,11pt]{article}
\usepackage{a4wide}
\usepackage[UKenglish]{babel,isodate}
\usepackage{amsmath,amssymb,amsthm,mathtools}
\usepackage[ebgaramond,lining,lf,nf]{newtx}

\usepackage{thm-restate}
\usepackage[svgnames,table]{xcolor}
\usepackage{bm}
\usepackage{mathrsfs}
\usepackage{optidef}
\usepackage{enumitem}
\usepackage{complexity}
\usepackage{pgfplots}
\pgfplotsset{compat=1.18}
\usepackage{hyperref} 
\hypersetup{colorlinks=true,citecolor=RoyalBlue,linkcolor=RoyalBlue,urlcolor=RoyalBlue}
\usepackage[capitalise]{cleveref}

\newclass{\UGC}{UGC}

\renewcommand\epsilon{\ensuremath{\varepsilon}}

\newcommand{\thresh}{\ensuremath{\mathcal{THRESH}}}
\newcommand{\rot}{\ensuremath{\mathcal{ROT}}}

\newcommand{\maxcut}{\textsc{Max-Cut}}
\newcommand{\cutorient}{\textsc{Cut-Orient}}

\newcommand{\maxdicut}{\textsc{Max-DiCut}}
\newcommand{\maxtwosat}{\textsc{Max-2Sat}}
\newcommand{\maxtwoand}{\textsc{Max-2And}}
\newcommand{\maxdicutcut}{\textsc{Max-DiCut-vs.-Cut}}
\newcommand{\maxcolour}{\textsc{Max-$k$-Col-vs.-$\ell$-Col}}

\DeclareMathOperator*{\Exp}{\mathbb{E}}
\DeclareMathOperator{\sgn}{sgn}
\DeclareMathOperator{\arccsc}{arccsc}
\DeclareMathOperator{\cut}{\textsc{Cut}}
\DeclareMathOperator{\dicut}{\textsc{DiCut}}

\renewcommand{\R}{\ensuremath{\mathbb{R}}}
\newcommand{\N}{\ensuremath{\mathbb{N}}}

\newcommand{\vt}{\ensuremath{\mathbf{x}_0}}

\newcommand{\va}{\ensuremath{\mathbf{a}}}
\newcommand{\vx}{\ensuremath{\mathbf{x}}}
\newcommand{\vy}{\ensuremath{\mathbf{y}}}
\newcommand{\vz}{\ensuremath{\mathbf{z}}}

\theoremstyle{plain}
\newtheorem{theorem}{Theorem}
\newtheorem{lemma}[theorem]{Lemma}

\newtheorem*{lemma*}{Lemma}

\newtheorem*{proposition*}{Proposition}

\newtheorem*{corollary*}{Corollary}

\theoremstyle{definition}
\newtheorem{definition}[theorem]{Definition}

\begin{document}

\author{Tamio-Vesa Nakajima\\
University of Oxford\\
\texttt{tamio-vesa.nakajima@cs.ox.ac.uk}
\and
Stanislav \v{Z}ivn\'y\\
University of Oxford\\
\texttt{standa.zivny@cs.ox.ac.uk}
}

\title{An approximation algorithm for Maximum DiCut vs.~Cut\thanks{This work was supported by UKRI EP/X024431/1 and by a Clarendon Fund Scholarship. For the purpose of Open Access, the authors have applied a CC BY public copyright licence to any Author Accepted Manuscript version arising from this submission. All data is provided in full in the results section of this paper.}}

\date{\today}
\maketitle

\begin{abstract}

Goemans and Williamson designed a 0.878-approximation algorithm for \maxcut{} in undirected graphs~[JACM'95]. Khot, Kindler, Mosel, and O'Donnel showed that the approximation ratio of the Goemans-Williamson algorithm is optimal assuming Khot's Unique Games Conjecture~[SICOMP'07]. In the problem of maximum cuts in \emph{directed} graphs (\maxdicut), in which we seek as many edges going from one particular side of the cut to the other, the situation is more complicated but the recent work of Brakensiek, Huang, Potechin, and Zwick showed that their 0.874-approximation algorithm is tight under the Unique Games Conjecture (up to a small delta)~[FOCS'23].

We consider a promise version of the problem and design an SDP-based algorithm which, if given a directed graph $G$ that has a directed cut of value $\rho$, finds an undirected cut in $G$ (ignoring edge directions) with value at least $\rho$.

\end{abstract}

\section{Introduction}

The maximum cut problem in \emph{undirected} graphs (\maxcut) is a fundamental
problem, seeking a partition of the vertex set into two
parts while maximising the number of edges going across. While \maxcut{} is
\NP-complete~\cite{Karp1972}, a random assignment leads to a $1/2$-approximation
algorithm. In their influential work, Goemans and Williamson gave the first
improvement and presented an SDP-based $\alpha_{\text{GW}}$-approximation
algorithm~\cite{GW95}, where $\alpha_{\text{GW}}\approx 0.878$.
Under Khot's Unique Games Conjecture (\UGC)~\cite{Khot02stoc}, this approximation factor is optimal~\cite{KKMO07,Mossel10:ann}. The currently best known inapproximability result not relying on \UGC{} is 
$16/17\approx 0.941$~\cite{Trevisan00:sicomp} (obtained by a gadget from H{\aa}stad's optimal inapproximability result~\cite{Hastad01}).

The maximum cut problem in \emph{directed} graphs (\maxdicut) is a closely related and well-studied \NP-complete problem, seeking a partition of the vertex set into two parts while maximising the number of edges going across in the prescribed direction. A random assignment leads to a $1/4$-approximation algorithm. In the first improvement over the random assignment, Goemans and Williamson presented an SDP-based $0.796$-approximation algorithm~\cite{GW95}. By considering a stronger SDP formulation (with triangle inequalities), Feige and Goemans later presented an $0.859$-approximation algorithm for \maxdicut~\cite{FG95}, building on the work of Feige and Lov\'asz~\cite{FL92}. Follow-up works by Matuura and Matsui~\cite{Matuura2003new} and by Lewin, Livnat, and Zwick~\cite{Lewin02:ipco} further improved the approximation factor. On the hardness side, the best inapproximability factor under \NP-hardness is $12/13\approx 0.923$~\cite{Trevisan00:sicomp} (again, via a gadget from a result in~\cite{Hastad01}).
In recent work, Brakensiek, Huang, Potechin, and Zwick gave an $0.874$-approximation algorithm for \maxdicut, also showing that the approximation factor is \UGC{}-optimal (up to a small delta)~\cite{Brakensiek23:focs}.

\paragraph*{Contributions}

Our main contribution is an SDP-based algorithm for \maxdicutcut{}. This algorithm achieves the best possible performance on the problem: given a directed graph $G$ that has a directed cut of value $\rho$, it finds a cut in $G$ (ignoring directions) of value at least $\rho$.\footnote{No better approximation ratio than~$1$ is possible. For example, given a directed graph $G$ with a directed cut of value $2 / 3$, it is not possible to find a cut of value greater than $2 / 3$ in general, since we might get the input graph $([5], \{ 1 \to 2, 3 \to 4, 5 \to 5 \})$, which doesn't have such a cut.}

For the following, we define $\dicut(G)$ to be the maximum value of any directed
cut in $G$, and $\cut(c)$ to be the value of a cut $c$ of $G$. (These and other notations will be defined fully in~\cref{sec:prelims}.) We prove our result in two steps. First we provide a randomised algorithm.

\begin{restatable}{theorem}{mainrand}\label{thm:mainrand}
    There is a randomised algorithm which, if given a directed graph $G$ and $\epsilon \in (0, 1)$, finds a cut $c$ in $G$ with expected value at least $\dicut(G) - \epsilon$
    in polynomial time in $|G|$ and $\log(1/\epsilon)$.
\end{restatable}

Second, we derandomise the algorithm from~\cref{thm:mainrand}.

\begin{restatable}{theorem}{mainunrand}\label{thm:mainunrand}
    There is a deterministic algorithm which, if given a directed graph $G$, finds 
    a cut $c$ of $G$ with $\cut(c) \geq \dicut(G)$ in polynomial time in $|G|$.
\end{restatable}

We note that the guarantee of the algorithm from~\cref{thm:mainunrand} does not have an
$\epsilon$ and is exact. Indeed, as the value of a cut is a number of the form $k / |E(G)|$, this can be achieved by
setting the error parameter in~\cref{thm:mainrand} small enough. 

\smallskip
Our algorithm in~\cref{thm:mainunrand} uses a \thresh{} rounding scheme (indeed, our rounding
scheme, as we shall discuss later, can  be seen as a distribution of \rot{} rounding
schemes --- these are a subfamily of the \thresh{} family). The \thresh{} rounding scheme was introduced in~\cite{Lewin02:ipco}, whereas the \rot{} rounding scheme was first introduced in~\cite{FG95}, but named in~\cite{Lewin02:ipco}.
They have been used ever since then; for instance,
recently to create an (almost) optimal
algorithm for \maxdicut{}~\cite{Brakensiek23:focs}.
That we use such a rounding scheme is not surprising: it is known by a slight generalisation of~\cite{Austrin10} established in~\cite{Brakensiek23:focs} that, assuming \UGC{} and Austrin's \emph{positivity conjecture}~\cite[Conjecture 1.3]{Austrin10}, the optimal algorithm for approximately solving \emph{any} Boolean 2-CSP uses such a rounding scheme; and this result can be easily extended to promise CSPs.
The challenge is specifying precisely \emph{which} rounding scheme in the
\thresh{} family to use: each member of this family is specified by an enormous amount of information. This family of rounding schemes is extremely
vast, including the state-of-the-art algorithms for, e.g., \maxcut{}~\cite{GW95}, \maxdicut{}~\cite{Brakensiek23:focs}, \maxtwosat{}~\cite{Lewin02:ipco} and \maxtwoand{}~\cite{Brakensiek23:focs}.
We remark that, whereas
many of the algorithms using such a rounding scheme have a numerical
flavour~\cite{Lewin02:ipco,Brakensiek23:focs}, i.e., their rounding scheme is
often given by a distribution over a set of functions that linearly interpolate
some points, our algorithm has a more analytic flavour. We also find it more
natural and self-contained to explain the rounding scheme by itself, rather than
as an instance of the \thresh{} family; thus we will take this direct,
self-contained approach.

We now give a brief, informal overview of our algorithm from~\cref{thm:mainrand}.
Given a directed graph $G = (V, E)$ with a directed cut of value $\rho$, we will try to place each vertex $u \in V$ at a point $\vx_u \in \R^N$, for some large enough $N$, such that $|| \vx_u ||_2^2 = 1$. In $\R^N$ we select a distinguished direction, given by a unit vector $\vt$. We will place these points so as to maximise an objective function that is a sum, over $(u, v) \in E$, of a linear combination between $|| \vx_u - \vx_v ||_2^2 = 2 - \vx_u \cdot \vx_v$ and $\vt (\vx_v - \vx_u)$; in addition the placement must satisfy some extra constraints known as the ``triangle inequalities''. This the standard SDP relaxation of the problem with triangle inequalities. The main novelty comes from our rounding procedure and its proof of correctness: if a vertex $u$ has $| \vx_u \cdot \vt |$ large, then we round it to $\sgn(\vx_u \cdot \vt)$; whereas if it has $|\vx_u \cdot \vt|$ small, we round it according to the sides of a uniformly random hyperplane (as in the Goemans-Williamson rounding~\cite{GW95}). 

The justification for this idea is the following. The vector $\vt$ represents the value $+1$ in our SDP.\@ Thus, if $\vx_u$ is very close to $\vt$ or $-\vt$ i.e., if $| \vx_u \cdot \vt|$ is large, then we ought to prioritise the $\vt$ component in $\vx_u$ and round according to $\sgn(\vx_u \cdot \vt)$. Conversely, consider an edge $(u, v)$ where $|\vx_u \cdot \vt|$ and $|\vx_v \cdot \vt|$ are both small. Then the contribution of $\vt (\vx_v - \vx_u)$ to the objective will likewise be small, and we care more about the contribution of $|| \vx_u - \vx_v||_2^2$. But this is precisely the contribution given by the vectors in the SDP relaxation for \maxcut{}; we intuitively want, therefore, to round as dictated by the best known algorithm~\cite{GW95}, i.e., to round according to a random hyperplane. In the case where we mix the sizes of $|\vx_u \cdot \vt|$ and $|\vx_v \cdot \vt|$, it can be seen that the contribution of the edge to the objective is approximately $1 / 2$. Also, the probability of properly cutting the edge is  $1 / 2$. We have not specified the meaning of ``large'' or ``small'' in this description: it turns out to be sufficient to simply select the threshold between ``large'' and ``small'' values uniformly at random in $[0, 1]$!

To help visualise this algorithm, we briefly explain it as a distribution of \rot{} rounding schemes. In such a rounding scheme, one specifies a \emph{rotation function} $f : [0, \pi] \to [0, \pi]$, and then considers every vector $\vx_u$. One then rotates $\vx_u$ to $\vx_u'$ within the plane spanned by $\vx_u$ and $\vt$ so that $\arccos(\vx_u' \cdot \vt) = f(\arccos(\vx_u' \cdot \vt))$, and finally rounds according to the sides of a random hyperplane. Our rounding scheme defines $f_a : [0, \pi] \to [0, \pi]$ as in~\cref{fig:rotation}, selects $a$ uniformly at random in $[0, 1]$, then rounds according to the \rot{} family of rounding schemes with rotation function $f_a$.
By analysing this rounding scheme carefully, we then determine that it has the advertised performance. We contain the key required bound of this analysis within~\cref{lem:bound}, stated and proved in~\cref{sec:analysis}.

\begin{figure}
\centering
\caption{Rotation function $f_a$}\label{fig:rotation}
\begin{tikzpicture}[>=stealth]
    \begin{axis}[
        xmin=-.2,xmax=pi+0.4,
        ymin=-.2,ymax=pi+0.4,
        axis x line=middle,
        axis y line=middle,
        axis line style=<->,
        xtick={0, 0.75, pi-.75, pi},
        xticklabels={0, $\arccos a$, $\pi-\arccos a$, $\pi$},
        ytick={0, pi/2, pi},
        yticklabels={0, $\pi / 2$, $\pi$},
        xlabel={$x$},
        ylabel={{\color{blue}$f_a(x)$}},
        grid=both
        ]
        \addplot[blue,very thick] coordinates{
        (0, 0)
        (.75, 0)
        };
        \addplot[blue,very thick,dashed] coordinates{
        (.75, 0)
        (.75, pi/2)
        };
        \addplot[blue,very thick] coordinates{
        (.75, pi/2)
        (pi-.75, pi/2)
        };
        \addplot[blue,very thick,dashed] coordinates{
        (pi-.75, pi/2)
        (pi-.75, pi)
        };
        \addplot[blue,very thick] coordinates{
        (pi-.75, pi)
        (pi, pi)
        };
    \end{axis}
\end{tikzpicture}
\end{figure}

We finish with mentioning a related (and in some sense dual) problem, which we call the \cutorient{} problem.
Suppose we are given an undirected graph $G$ which has a maximum cut of value
$\rho$. The problem asks us to orient the edges of $G$ so that the resulting
oriented graph has a directed cut of as large a value as possible. Observe that
this problem has a simple algorithm which produces an oriented graph with a
directed cut of size $\alpha_{\text{GW}} \rho$: find a cut in $G$ of size
$\alpha_{\text{GW}} \rho$ using the Goemans-Williamson algorithm, then orient
the edges according to that cut. Our deterministic algorithm for \maxdicutcut{}
shows that, assuming \UGC{}, nothing better is possible as our algorithm reduces \maxcut{} to \cutorient{}.

\paragraph*{Related work}

\maxcut{} is one of the simplest examples of a Constraint Satisfaction Problems (CSPs)~\cite{KSTW00}. A  celebrated result of Raghavendra showed that the basic SDP relaxation gives an optimal approximation algorithm to every CSP (assuming \UGC{})~\cite{Raghavendra08:everycsp}. While a breakthrough result, there are some caveats as as already discusses earlier in this Introduction and also discussed in more detail, e.g., in~\cite{MakarychevM17,Brakensiek23:focs}.
In particular, the universal rounding scheme of Raghavendra can only go arbitrarily close to the integrality gap (in polynomial time), not match it exactly --- thus it would only be able to have an approximation ratio of $1 - \epsilon$ for any fixed $\epsilon$, not exactly 1 as in our algorithm here. The exact performance of this rounding scheme is also nontrivial to find e.g.~separating the optimal approximation ratio for \maxcut{} and \maxdicut{} was only recently shown in~\cite{Brakensiek23:focs}.
Indeed, the algorithm of Raghavendra and Steurer~\cite{Raghavendra09:focs} computes a gap instance for an SDP to error $\epsilon$ in \emph{doubly exponential time in $\poly(1 / \epsilon)$}.

The studied \maxdicutcut{} problem falls within the framework of so-called valued promise CSPs, recently investigated systematically by Barto et al.~\cite{bbkvz24:arxiv}. Another valued promise CSP recently studied is the \maxcolour{} problem~\cite{nz23:arxiv}, in which one seeks an $\ell$-colouring of a given graph $G$ with as many edges coloured properly as possible given that a $k$-colouring of $G$ exists with many edges coloured properly.

\section{Preliminaries}\label{sec:prelims}

\paragraph*{Notation} We use the Iverson bracket $[\phi]$, which is 1 when $\phi$ is true and 0 otherwise. We write $[n] = \{1, \ldots, n\}$. 
Whenever we write an expectation such as $\Exp_{x \in X}[f(x)]$, we implicitly mean that $x$ is uniformly distributed in $X$. 

We shall use two trigonometric functions: the cosecant function, defined by
$\csc(x) = 1 / \sin(x)$, and its inverse, denoted by $\arccsc$. Recall that
$\arccsc$ is decreasing and, for $x \geq 1$, we have $0 < \arccsc(x) \leq \pi / 2$.

A \emph{hyperplane} is a function $H : \R^n \to \{ -1, 0, 1 \}$ of the form
$H(\vx) = \sgn(\vx \cdot \va)$.
A \emph{uniformly random} hyperplane is one in which $\va$ is selected from a rotationally symmetric distribution uniformly at random e.g.~from the standard multivariate normal distribution. It is well known (cf.~\cite{GW95})
that $\Pr_{H}[ H(\vx) \neq H(\vy) ] = (1 / \pi) \arccos(\vx \cdot \vy)$ and that $\Pr_H[ H(\vx) = 0 ] = 0$. It is not difficult to see that this is the case even if we condition that, for some fixed vector $\vt$, we have $H(\vt) = 1$ (this just standardises which side of $H$ is $+1$ and which is $-1$).

For a vector $\vx \in \mathbb{R}^N$, we let $\langle \vx \rangle \subseteq \mathbb{R}^N$ denote the subspace spanned by $\vx$. For a linear subspace $S \subseteq \mathbb{R}^N$, we let $S^\perp$ denote the subspace of vectors perpendicular to $S$.

\paragraph*{Problem definitions} Consider a directed graph $G = (V, E)$. A
\emph{cut} in  $G$ is a partition of $V$ into two parts $A, B$. The \emph{undirected value} of this cut is $| (A \times B \cup B \times A) \cap E | / |E|$. The \emph{directed value} of the cut is $|(A \times B) \cap E| / |E|$. In other words, the directed value of a cut is the proportion of edges from $A$ to $B$, and the undirected value of a cut is the proportion of edges between $A$ and $B$ ignoring orientation.

While these definitions are the most natural, it will be  more convenient (and is standard in the literature) to work with the following equivalent ones instead. We define, for $x, y \in \R$,
\begin{align*}
\cut(x, y) & = \frac{1 - xy}{2}, \\
\dicut(x, y) & = \frac{1 - xy + y - x}{4}.
\end{align*}
For a graph $G = (V, E)$, a cut is a function $c : V \to \{ \pm 1 \}$. The undirected value of $c$ is
\[
\cut(c) = \Exp_{(u, v) \in E} [ \cut(c(u), c(v))] = \frac{1}{2} \Exp_{(u, v) \in E}[1 - c(u) c(v)].
\]
The directed value of $c$ is
\[
\dicut(c) = \Exp_{(u, v) \in E}[ \dicut(c(u), c(v))] = \frac{1}{4} \Exp_{(u, v) \in E}[1 - c(u) c(v) + c(v) - c(u)].
\]

Now, we define the (directed/undirected) value of a directed graph $G$ as follows
\begin{align*}
    \cut(G) & = \max_{c : V \to \{ \pm 1\} } \cut(c). \\
    \dicut(G) & = \max_{c : V \to \{ \pm 1\} } \dicut(c).
\end{align*}

The \maxcut{} and \maxdicut{} problems are simple to describe: given a directed
graph $G$, find or approximate $\cut(G)$ or $\dicut(G)$ respectively.\footnote{In the way we have defined our problems, the instance of \maxcut{} is a
directed graph, but the directions are arbitrary and do not matter for the
problem.} The problem we are interested in is a promise problem bridging the gap between the two, which we now define.

\begin{definition}
    The \maxdicutcut{} problem is defined as follows. Given a directed graph $G = (V, E)$, for which $\dicut(G) = \rho$, find a cut $c : V \to \{ \pm 1 \}$ such that $\cut(c) \geq \rho$.
\end{definition}

We observe that such a cut must exist; indeed, a cut $c$ witnessing  $\dicut(G) = \rho$
works as $\dicut(c) \leq \cut(c)$ for every cut $c$.

\paragraph*{Semidefinite programming}

We derive the basic SDP relaxation for \maxdicut{}, including triangle constraints. This SDP has been used many times before, e.g.~in~\cite{Brakensiek23:focs}, and was first used in~\cite{FG95}.

\maxdicut{} can be formulated as a quadratic program, as follows. We introduce variables $x_v \in \R$ for $v \in V$; then we solve the following maximisation problem.

\begin{maxi*}|l|
{x_v}{
\Exp_{(u, v) \in E}\left[\dicut(x_u, x_v)\right]
= \Exp_{(u, v) \in E}\left[\frac{1 - x_u x_v + x_v - x_u}{4}\right]}
{}{}
\addConstraint{x_v^2 = 1\qquad}{}{\forall v \in V}
\end{maxi*}

We change this program in two ways. First, we add the so-called triangle
constraints (which hold for $x_v \in \{ \pm 1\}$), to get the following:

\begin{maxi}|l|
{x_v}{
\Exp_{(u, v) \in E}\left[\frac{1 - x_u x_v + x_v - x_u}{4}\right]}
{}{}\label{qp1}
\addConstraint{x_v^2 = 1}{}{\forall v \in V}
\addConstraint{(1 \pm x_u)(1 \pm x_v) \geq 0\qquad}{}{\forall (u, v) \in E}
\end{maxi}

The main problem with this formulation is that it is not homogeneous i.e.~there are linear terms. In order to fix this, we introduce a variable $x_0 \in \{ \pm 1 \}$,\footnote{Assume without loss of generality that $0 \not \in V$.} with which we multiply all linear factors. The solutions to this new program will correspond to the solutions to the old program~\eqref{qp1}, where the correspondence is given by dividing by $x_0$ or by setting $x_0 = 1$.

\begin{maxi}|l|
{x_v, t}{
\Exp_{(u, v) \in E}\left[\frac{1 - x_u x_v + x_0 x_v - x_0 x_u}{4}\right]}
{}{}\label{qp2}
\addConstraint{x_v^2= 1}{}{\forall v \in V \cup \{ 0 \}}
\addConstraint{(x_0 \pm x_u) (x_0 \pm x_v) \geq 0\qquad}{}{\forall (u, v) \in V}
\end{maxi}

Unfortunately,~\eqref{qp2} is a quadratic programming problem, which is \NP-hard to solve in general (and in particular, solving this SDP would solve \maxdicut{}, which is not possible assuming $\P \neq \NP$). Thus we
get the canonical SDP relaxation by allowing $x_u$ to be vectors in an
arbitrarily high-dimensional vector space, replacing products between real
numbers with inner products. From the way we have derived this program it is clear that its value is at least as large as $\dicut(G)$.

\begin{maxi}|l|
{\vx_v, \vt}{
\Exp_{(u, v) \in E}\left[\frac{1 - \vx_u \cdot \vx_v + \vt \cdot \vx_v - \vt \cdot \vx_u}{4}\right]}
{}{}\label{sdp1}
\addConstraint{|| \vx_v ||^2 = 1}{}{\forall v \in V \cup \{ 0 \}}
\addConstraint{(\vt \pm \vx_u) \cdot (\vt \pm \vx_v) \geq 0\qquad}{}{\forall (u, v) \in E}
\end{maxi}

This semidefinite program can be solved with error at most $\epsilon \in (0, 1)$ in
polynomial time in $|G|$ and $\log(1/\epsilon)$, either by the ellipsoid method~\cite{GrotschelLS81:combinatorica} or by interior point methods~\cite{KlerkV16:siamjo}.\footnote{For simplicity we will ignore issues of real precision.}

\section{Results}

In this section, we prove our main results.

\mainrand*
\mainunrand*

In order to prove this, we will need the following technical lemma, which is proved in~\cref{sec:analysis}.

\begin{restatable}[Configuration lemma]{lemma}{config}\label{lem:config}
Suppose we have three unit vectors $\vx_u, \vx_v, \vt \in \R^N$ such that $(\vt \pm \vx_u) \cdot (\vt \pm \vx_v) \geq 0$. Let $\vy_u, \vy_v$ be the normalised projection of $\vx_u, \vx_v$ onto $\langle \vt \rangle^\perp$ (if $\vx_w \in \langle \vt \rangle$ then set $\vy_w$ to be some vector perpendicular to all other vectors in the problem). Suppose $a$ is selected uniformly at random from $[0, 1]$, and $H$ is a hyperplane selected uniformly at random such that $H(\vt) = 1$. Suppose $c(u), c(v)$ are defined as follows: $c(w) = 1$ if $\vx_w \cdot \vt \geq a$ or if $\vx_w\cdot \vt > -a$ and $H(\vy_w) = 1$; otherwise $c(w) = -1$. Then, we have
\[
\Pr_{a, H}[ c(u) \neq c(v) ] \geq \frac{1 - \vx_u \cdot \vx_v + \vt \cdot \vx_v - \vt \cdot \vx_u}{4}.
\]
\end{restatable}

\begin{proof}[Proof of~\cref{thm:mainrand}]
    Suppose we are given an input graph $G$ with $\dicut(G) = \rho$. Then, we
    first solve the semidefinite program~\eqref{sdp1}, with error $\epsilon$;
    suppose that the solution vectors we get are $\vx_u \in \R^{N}$ for $u \in V \cup \{0 \}$, where $N$ can be taken to be equal to the number of vectors in the problem i.e.~$|V| + 1$.\footnote{Indeed, the SDP algorithm actually computes the inner products between all vectors $\vx_0, \vx_v$ --- the actual vectors can be reconstituted using a Cholesky decomposition i.e.~writing a positive semidefinite matrix $A \in \mathbb{R}^{n \times n}$ as a product $X^T X$ for $X \in \mathbb{R}^{n \times n}$. If $A$ is the matrix containing the inner products, then the columns of $X$ form the vectors $\vx_0, \vx_v$ in the solution.}
    The value of these vectors is at least $\rho - \epsilon$. Now, we randomly round
    this SDP solution into a solution to our original problem. We first intuitively explain our rounding scheme.

    Observe that the vector $\vt$ corresponds to $+1$ in our problem, whereas the vector $-\vt$ corresponds to $-1$. If some vector $\vx_u$ were equal exactly to $\vt$ or $-\vt$, then intuitively we should round it to $\pm 1$ accordingly. Now, similarly, if $| \vx_u \cdot \vt|$ is large enough, then we should also round $\vx_u$ according to $\sgn(\vx_u \cdot \vt)$. On the other hand, if $| \vx_u \cdot \vt |$ is small enough, then we may as well throw away the component of $\vx_u$ parallel to $\vt$. Thus, we set $\vy_u$ to be the normalised projection of $\vx_u$ onto $\langle \vt \rangle^\perp$. In the case $\vx_u \in \langle \vt \rangle$, then of course $| \vx_u \cdot \vt |$ is large, so the value we assign $\vy_u$ shouldn't matter; we set it to some vector perpendicular to all other vectors seen so far for simplicity. (This requires increasing the dimension of all the vectors by $|V|$, but this does not affect the complexity.)

    Consider now an edge $(u, v)$ where both $| \vx_u \cdot \vt|$ and $| \vx_v \cdot \vt |$ are small. In this case, we see that $\vx_v \cdot \vt - \vx_u \cdot \vt$ does not contribute much to the objective function. Indeed, the major contribution is by $1 - \vx_u \cdot \vx_v = 2\cut(\vx_u, \vx_v)$. Thus, it intuitively makes sense to round these vectors by a random hyperplane.

    We are now ready to describe precisely the operation of the algorithm. First we sample a threshold $a \in [0, 1]$ uniformly at random. (This threshold separates ``large'' from ``small'' values in the informal discussion above.) If $| \vx_u \cdot \vt | \geq a$ then set $c(u) = \sgn(\vx_u \cdot \vt)$. For all $u$, compute $\vy_u$ as described above, namely set
    \[
    \vy_u = \frac{\vx_u - (\vx_u \cdot t) \vt}{\sqrt{1 - (\vx_u \cdot \vt)^2}}
    \]
    whenever $|\vx_u \cdot \vt| < 1$, and otherwise $\vy_u$ is set perpendicular to all other vectors seen so far. Sample a hyperplane $H$ uniformly at random such that $H(\vt) = 1$. Now, for all $u$ for which $|\vx_u \cdot \vt| < a$, we set $c(u) = H(\vy_u)$.

    For any edge $(u, v) \in E$, what is the probability that this edge is correctly cut by our method? We observe that the rounding method we described is exactly the same as the one mentioned in~\cref{lem:config}, hence
    \[
    \Pr_{a, H}[c(u) \neq c(v)] \geq \frac{1 - \vx_u \cdot \vx_v + \vt \cdot \vx_v - \vt\cdot\vx_u}{4}.
    \]
    Thus, noting that $\Exp_{a, H}[\cut(c(u), c(v))] = \Pr_{a, H}[c(u) \neq c(v)]$, we have
    \[
    \Exp_{a, H}[ \cut(c) ] \geq \Exp_{(u, v) \in E}\left[
    \frac{1 - \vx_u \cdot \vx_v + \vt \cdot \vx_v - \vt\cdot\vx_u}{4} \right] \geq \rho - \epsilon,
    \]
    as required.
\end{proof}

We now derandomise our algorithm. This derandomisation scheme essentially uses the fact that our rounding scheme is a family of \rot{} rounding schemes i.e.~schemes that essentially rotate our vectors by some angle towards/away from $\vt$.

\begin{proof}[Proof of~\cref{thm:mainunrand}]
    We first slightly modify our algorithm. For every vector $\vx_u$ where $|\vx_u \cdot \vt| < a$ we set $\vz_u = \vy_u$; otherwise we set $\vz_u = \sgn(\vx_u \cdot \vt) \vt$.

    For any hyperplane where $H(\vt) = 1$ and any $u \in V$ where $|\vx_u \cdot t| \geq a$, we have that $H(\vz_u) = \sgn(\vx_u \cdot \vt)$. Thus, the following algorithm has the exact same characteristics as the algorithm from~\cref{thm:mainunrand}: compute an SDP solution $\vt, \vx_u$, then select a uniformly random threshold $a \in [0, 1]$ and compute $\vy_u, \vz_u$ as described. Now, select a uniformly random hyperplane $H$, with $H(\vt) = 1$. Finally, set $c(u) = H(\vz_u)$. This algorithm will in fact output exactly the same cut $c$ as our original one as long as the same values of $a, H$ are selected.

    Now, observe that the performance bound on our original algorithm says that
    the expected value of a cut $c$ created by selecting $a$ uniformly at random from $[0, 1]$ and setting $c(u) = H(\vz_u)$ is at
    least $\rho - \epsilon$. First, we observe that the exact value of $a$ doesn't matter, only its value relative to $|\vx_u \cdot \vt|$ for $u \in V$. Thus there are essentially $O(|V|)$ choices for $a$, and we can try all of them. Now, for the best $a$, we know that the expected value of the cut given by $c(u) = H(\vz_u)$ is at least $\rho - \epsilon$.
    It is well known that given such vectors $\vz_u$
    we can compute a colouring of $G$ with value at least $\rho - \epsilon -
    \epsilon'$ in polynomial time in $|G|$ and $\log(1/\epsilon')$ for $\epsilon \in (0, 1)$ --- indeed this amounts to derandomising the original Goemans-Williamson \maxcut{} algorithm. One classic algorithm for this is given in~\cite{Mahajan99:sicomp}, whereas a simpler and more optimised algorithm is given in~\cite{BK05}.

    Thus, by setting $\epsilon + \epsilon' = 1 / 2|E|$, we get a cut $c$ with
    value at least $\rho - 1 / 2|E| = \dicut(G) - 1 / 2|E|$. But observe that
    $\dicut(G) = k / |E|$, where $k$ is the number of edges in the optimal
    directed cut in $G$, and also $\cut(c) = k' / |E|$, where $k'$ is the number of edges properly cut by $c$. It follows therefore, as $k, k'\in \N$ and $k' / |E| \geq k / |E| - 1 / 2|E|$ i.e.~$k' \geq k - 1/2$ that $k' \geq k$, and hence $\cut(c) \geq \dicut(G)$.
\end{proof}

\section{Analysis}\label{sec:analysis}

In this section we prove the following lemma.

\config*

We define the value $\rho_{xyz}$. When $x = \pm 1$ or $y = \pm 1$, we define $\rho_{xyz} = 0$; otherwise,
\[
\rho_{xyz} = \frac{z - xy}{\sqrt{1 - x^2}\sqrt{1 - y^2}}.
\]
In order to prove~\cref{lem:config}, we will need the following fact.

\begin{lemma}\label{lem:bound}
    Consider $x, y \in [-1, 1]$ with $|x| \leq y$, and take $z \in [ x + y - 1, 1 - y + x ]$. Then
    \begin{equation}\label{eq:bound}
    \frac{1 - y}{\pi} \arccos \rho_{xyz} \geq \frac{1 - z + x - y}{4}.
    \end{equation}
\end{lemma}

To prove this fact we define the following functions:
\begin{align*}
    F(x, y, z) & = \arccos \rho_{xyz} - \frac{\pi}{4} \frac{1 - z + x - y}{1 - y}, \\
    \Delta_{xy} & = 4 \pi^4 (1 - x^2)(1 - y^2) - 64 \pi^2 (1 - y)^2, \\
    \alpha_{xy} & = \frac{\pi}{4} \frac{\sqrt{1 - x^2}\sqrt{1 - y^2}}{1 - y}, \\
    G(x, y) &= \pi - \arccsc \alpha_{xy} - \frac{\pi}{4}(1 + x) - \sqrt{\alpha_{xy}^2 - 1}.
\end{align*}

For the function $G(x, y)$ to be defined (over the reals), it is necessary for $|\alpha_{xy}| \geq 1$;  in particular this is certainly true when $\Delta_{xy} \geq 0$, which is the only case when we will use this function. For many other of these functions, we require $y < 1$; this will also always be the case when they are invoked.

\begin{lemma}\label{lem:inexy}
    Assume $x, y \in (-1, 1)$, $y \geq |x|$, $\Delta_{xy} \geq 0$ and $xy - \frac{\sqrt{\Delta_{xy}}}{2\pi^2} \geq x + y - 1$. Then $G(x, y) \geq 0$.
\end{lemma}
\begin{proof}
    We observe that for $xy - \frac{\sqrt{\Delta_{xy}}}{2\pi^2} \geq x + y - 1$ to be true, we must have
    \begin{multline*}
    0 \leq 1 + xy - x - y - \frac{\sqrt{\Delta_{xy}}}{2\pi^2} = (1 - x)(1 - y) - \frac{\sqrt{\Delta_{xy}}}{2\pi^2} = \\
    \frac{1}{(1-x)(1-y) + \frac{\sqrt{\Delta_{xy}}}{2\pi^2}}\left({(1 - x)^2(1-y)^2 - \frac{\Delta_{xy}}{4\pi^4} }\right) = \\
    \frac{2(1 - y)}{\pi^2(1-x)(1-y) + \frac{\sqrt{\Delta_{xy}}}{2}}\left(8(1 - y) - \pi^2(1 - x)(x + y)\right).
    \end{multline*}
    This is true if and only if $8(1 - y) \geq \pi^2(1 - x)(x + y)$. Thus $y \leq (\pi^2 x - \pi^2x^2 - 8)/(\pi^2 x -8 - \pi^2)$. 
    
    We observe now that $\alpha_{xy}$ is an increasing function of $y$ on $(-1, 1)$, hence by taking $y = (\pi^2 x - \pi^2 x^2 - 8)/(\pi^2 x -8 - \pi^2)$, we find that
    \[
    \alpha_{xy} \leq \sqrt{1 + \frac{\pi^2}{16}(1 - x)^2}.
    \]
    Furthermore, note that
    \[
    \frac{\partial}{\partial \alpha_{xy}} G(x, y) = - \frac{\sqrt{\alpha_{xy}^2 - 1}}{\alpha_{xy}} < 0.
    \]
    Hence we minimise $G$ by selecting $y$ so as to maximise $\alpha_{xy}$. Thus, taking $y$ so that $\alpha_{xy} = \sqrt{1 + \frac{\pi^2}{16}(1 - x)^2}$, we find that
    \begin{multline*}
    G(x, y) \geq \pi - \arccsc \left( \sqrt{1 + \frac{\pi^2}{16}(1 - x)^2}\right) - \frac{\pi}{4}(1 + x) - \sqrt{\sqrt{1 + \frac{\pi^2}{16}(1 - x)^2}^2 - 1} = \\
    \pi -
    \arccsc \left( \sqrt{1 + \frac{\pi^2}{16}(1 - x)^2}\right) 
     - \frac{\pi}{2} \geq 0,
    \end{multline*}
    as required.
\end{proof}

\begin{lemma}\label{lem:subst1}
    Assume $x, y \in (-1, 1)$ and $\Delta_{xy} \geq 0$. If $z = xy - \frac{\sqrt{\Delta_{xy}}}{2 \pi^2}$, then $\arccos \rho_{xyz} = \pi - \arccsc \alpha_{xy}$.
\end{lemma}

\begin{proof}
    Observe that $z - xy = - \frac{\sqrt{\Delta_{xy}}}{2\pi^2}$. So%
    \begin{multline*}
    \rho_{xyz}
    = -\frac{\sqrt{\Delta_{xy}}}{2\pi^2 \sqrt{1 - x^2} \sqrt{1 - y^2}}
    = -\frac{1}{2\pi^2} \sqrt{\frac{4\pi^4(1 - x^2)(1 - y^2) - 64\pi^2(1 - y)^2}{(1 - x^2)(1 - y^2)}} =\\
    -\frac{1}{2\pi^2}
    \sqrt{4 \pi^4 - 64\pi^2 \frac{(1 - y)^2}{(1 - x^2)(1 - y^2)}}
    = -\sqrt{1 - \frac{16}{
    \pi^2} \frac{(1 - y)^2}{(1 - x^2)(1 - y^2)}} = -\sqrt{1 - 1 / \alpha_{xy}^2}.
    \end{multline*}
    Thus recalling the trigonometric identity $\arccos (- \sqrt{1 - 1 / \alpha^2}) + \arccsc \alpha = \pi$, we have
    \[
    \arccos \rho_{xyz} = \arccos( - \sqrt{1 - 1 / \alpha_{xy}^2}) = \pi - \arccsc \alpha_{xy}.\qedhere
    \]
\end{proof}

\begin{lemma}\label{lem:subst2}
    Assume $x, y \in (-1, 1)$ and $\Delta_{xy} \geq 0$. If $z = xy - \frac{\sqrt{\Delta_{xy}}}{2\pi^2}$, then 
    \[
    \frac{\pi}{4} \frac{1 - z + x - y}{1 - y} = \frac{\pi}{4} (1 + x) + \sqrt{\alpha^2_{xy} - 1}.
    \]
\end{lemma}
\begin{proof}
We have
\begin{multline*}
    \frac{\pi}{4} \frac{1 - z + x - y}{1 - y}
    =
    \frac{\pi}{4} \frac{1 - xy + x - y}{1 - y} + \frac{1}{8\pi(1 - y)} \sqrt{4\pi^4(1 - x^2)(1-y^2) - 64\pi^2(1-y)^2} = \\
    \frac{\pi}{4}\frac{(1 + x)(1 - y)}{1 - y} + \sqrt{\frac{\pi^2}{16}\frac{(1 - x^2)(1- y^2)}{(1 - y)^2} - 1} = \frac{\pi}{4}(1 + x)  + \sqrt{\alpha_{xy}^2 - 1},
\end{multline*}
as required.
\end{proof}

\begin{lemma}\label{lem:inexyz}
    Take $x, y \in (-1, 1)$ such that $|x| \leq y$, and $x + y - 1 \leq z \leq 1 - y + x$. Then $F(x, y, z) \geq 0$.
\end{lemma}

For the purpose of illustration, we plot $F(x, y, z)$ for $x = 0.14, y = 0.36$ and $z \in [x + y - 1, 1 - y + x]$ in~\cref{fig:F}; these values for $x$ and $y$ have been chosen to highlight the more difficult case in our proof i.e.~when there is a local minimum. We observe that the triangle inequalities (i.e.~$z \in [x + y - 1, 1 - y + x]$) seem to be necessary for our rounding scheme to work; for example, $x = 0.03, y = 0.99, z = -0.01$ leads to $F(x, y, z) \approx -2.07$.

\begin{figure}
\centering
\caption{Plot of $F(0.14, 0.36, z)$}\label{fig:F}
\begin{tikzpicture}[>=stealth]
    \begin{axis}[
        width=\textwidth,
        height=\textwidth/2,
        xmin=-.5-.03,
        xmax=.78+.03,
        ymin=0.63-0.01,
        ymax=0.71+0.02,
        axis x line=middle,
        axis y line=middle,
        axis line style=<->,
        xlabel={$z$},
        ylabel={{\color{blue}$F(0.14, 0.36, z)$}},
        grid=both
        ]
        \addplot[no marks,blue,thick] expression[domain=-.5:.78,samples=300]{
        -0.957204 + 1.22718 * x + rad(acos(-0.0545594 + 1.08253 * x))};
    \end{axis}
\end{tikzpicture}
\end{figure}

\begin{proof}
    Fixing $x, y$, we want to find the $z$ that minimises the value of $F(x, y, z)$. We consider the extreme values of $z$, then inspect the local minima.
    \begin{description}
    \item[Case 1: $\bm{z = x + y - 1}$.] In this case, $1 - z + x - y = 2 - 2y$, so
    \[
    \frac{\pi}{4} \frac{1 - z + x - y}{1 - y} = \frac{\pi}{2}.
    \]
    Furthermore, $z - xy = -1 + x + y - xy = - (1 - x)(1 - y)$, so
    \[
    \rho_{xyz} = - \frac{(1 - x)( 1 - y)}{\sqrt{1 - x^2}\sqrt{1 - y^2}} < 0.
    \]
    Thus in this case $\arccos \rho_{xyz} > \pi / 2$, so $F(x, y, z) > 0$.
    \item[Case 2: $\bm{z = 1 - y + x}$.] In this case $1 - z + x - y = 0$, so
    \[
    F(x, y, z) = \arccos \rho_{xyz} - 0 \geq 0.
    \]
    \item[Case 3: $\bm{x + y - 1 < z < 1 - y + x}$.] We are now looking for local minima of $F(x, y, z)$ with respect to $z$. Thus, we compute
    \[
    \frac{\partial}{\partial z} F(x, y, z) = \frac{\pi^2(1 - x^2 -y^2 - z^2 + 2xyz) - 16(1 - y)^2}{4(1 - y)\sqrt{1- x^2 - y^2 - z^2 + 2xyz}(\pi \sqrt{1 - x^2 - y^2 - z^2 + 2xyz} + 4(1 - y))}.
    \]
    We observe that the denominator of this fraction is always positive, and thus the sign of $\frac{\partial}{\partial z} F(x, y, z)$ is given be the numerator, namely
    \[
    \pi^2(1 - x^2 -y^2 - z^2 + 2xyz) - 16(1 - y)^2.
    \]
    We compute the discriminant of this quadratic, and observe that it equals $\Delta_{xy}$. Thus, if $\Delta_{xy} < 0$ there are no stationary points; if $\Delta_{xy} = 0$ then there is one stationary point but it is neither a local minimum nor a local maximum; and if $\Delta_{xy} > 0$ there are two stationary points, located at
    \[
    z = xy \pm \frac{\sqrt{\Delta_{xy}}}{2\pi^2}.
    \]
    The first of these is a local minimum, and the second is a local maximum. Thus we are interested in $z = xy - \frac{\sqrt{\Delta_{xy}}}{2\pi^2}$. Of course if this value does not satisfy our boundary conditions, then there are no local minima; so we can assume that $xy - \frac{\sqrt{\Delta_{xy}}}{2\pi^2} \geq x + y - 1$. By~\cref{lem:subst1,lem:subst2}, when $z = xy - \frac{\sqrt{\Delta_{xy}}}{2\pi^2}$, we have
    \[
        F(x, y, z) = \pi - \arccsc \alpha_{xy} - \frac{\pi}{4}(1 + x) - \sqrt{\alpha^2_{xy} - 1} = G(x, y).
    \]
    By~\cref{lem:inexy}, $G(x, y) \geq 0$, so in this case it follows that $F(x, y, z) \geq 0$ also.
    \end{description}
    Thus, our conclusion follows in all cases.
\end{proof}

\begin{proof}[Proof of~\cref{lem:bound}]
    We first deal with the edge cases, when $x, y = \pm 1$, then with the typical case. As $y \geq |x| \geq 0$, for either $x$ or $y$ to be $\pm 1$, it must be the case that $y = 1$, so suppose this is the case. Then the left side of~\eqref{eq:bound} is $0$, whereas the right side is $(x - z) / 4$. By assumption $x + y - 1 \leq z$, but since $y = 1$ this is the same as $x \leq z$, hence $(x - z) / 4 \leq 0$ as required. We can thus assume that $y < 1$ and thus $x, y \in (-1, 1)$.
    In this case,~\eqref{eq:bound} is equivalent to $F(x, y, z) \geq 0$, which is implied by~\cref{lem:inexyz}.
\end{proof}

Having proved~\cref{lem:bound}, we return to the proof of~\cref{lem:config}.

\begin{proof}[Proof of~\cref{lem:config}]
    Suppose we let $x = \vx_u \cdot \vt$, $y = \vx_v \cdot \vt$ and $z = \vx_u \cdot \vx_v$. Then $(\vt \pm \vx_u) \cdot (\vt \pm \vx_v) \geq 0$ amounts to
    \begin{align*}
        1 + x + y + z & \geq 0, \\
        1 - x - y + z & \geq 0, \\
        1 - x + y - z & \geq 0, \\
        1 + x - y - z & \geq 0,
    \end{align*}
    or equivalently $|x + y| - 1 \leq z \leq 1 - |y - x|$. With this reformulation, we must show that
    \begin{equation}\label{eq:goal}
    \Pr_{a, H}[c(u) \neq c(v)] \geq \frac{1 - z + y - x}{4}.
    \end{equation}

    We first consider the effect of swapping $u$ and $v$, and/or negating $\vt$ has on the problem. Observe that all of these transformations maintain the requirement $(\vt \pm \vx_u) \cdot (\vt \pm \vx_v) \geq 0$. These transformations do not change $z$, but they have the effect of mapping $(x, y) \mapsto (y, x)$ (for swapping) and $(x, y) \mapsto (-x, -y)$ (for negating). Thus, by swapping we can achieve $(x, y) \mapsto (y, x)$, and by swapping and negating we can achieve $(x, y) \mapsto (-y, -x)$.
    
    We claim that these transformations mean we can consider only the case $x \leq y$, $|x| \leq |y|$. Indeed, if $x > y$ then by swapping $\vx_u$ and $\vx_v$ we can reduce to the case $x \leq y$ while increasing the right side of~\eqref{eq:goal} (i.e.~reducing to a necessary subcase). Furthermore, once we are in this case, if $|x| > |y|$, then by both swapping and negating we can get to the case $|x| \leq |y|$ while remaining in the case $x \leq y$ (and not changing the values in \eqref{eq:goal}). Henceforth we assume that $x \leq y$ and $|x| \leq |y|$.

    Observe that $x \leq y$ and $|x| \leq |y|$ are equivalent to $|x| \leq y$; hence we know that $y - x \geq 0$ and $x + y \geq 0$, so $|x + y| - 1 \leq z \leq 1 - |y - x|$ becomes $x + y - 1 \leq z \leq 1 - y + x$.

    Now, what is the inner product of $\vy_u$ and $\vy_v$? We observe that, as long as $| \vx_w \cdot \vt | < 1$, we have
    \[
    \vy_w = \frac{\vx_w - (\vx_w \cdot \vt) \vt}{\sqrt{1 - (\vx_w \cdot \vt)^2}}.
    \]
    Hence in this case we have
    \begin{multline*}
    \vy_u \cdot \vy_v = \frac{(\vx_u - (\vx_u \cdot \vt) \vt) (\vx_v - (\vx_v \cdot \vt) \vt)}{\sqrt{1 - (\vx_u \cdot \vt)^2}\sqrt{1 - (\vx_v \cdot \vt)^2}}
    = \frac{\vx_u \cdot \vx_v - (\vx_u \cdot \vt) (\vx_v \cdot \vt)}{\sqrt{1 - (\vx_u \cdot \vt)^2}\sqrt{1 - (\vx_v \cdot \vt)^2}} =\\
    \frac{z - xy}{\sqrt{1 - x^2}\sqrt{1 - y^2}}.
    \end{multline*}
    Conversely if $ | \vx_w \cdot t| = 1$, then $\vy_w$ is perpendicular to all other vectors $\vy_{w'}$, so we have $\vy_u \cdot \vy_v = 0$. Hence we observe that in all cases $\vy_u \cdot \vy_v = \rho_{xyz}$.

    Now, let us investigate what $\Pr_{a, H}[c(u) \neq c(v)]$ is. There are two main cases, depending on the sign of $x$.
    \begin{description}
        \item[Case 1: $\bm{x \leq 0}$.] In this case, with probability $-x$ we have that $a \leq |x| \leq |y|$, and hence $c(u) = -1, c(v) = 1$. With probability $|y| - |x| = x + y$, we have $|x| < a \leq |y|$, in which case $c(v) = 1$ but $c(u)$ is uniformly randomly distributed in $\{ \pm 1\}$. With probability $1 - |y| = 1 - y$ we have $a > |y| \geq |x|$, in which case $\vy_u, \vy_v$ are separated by the random hyperplane $H$ with probability $(1 / \pi) \arccos \rho_{xyz}$. So
        \[
        \Pr_{a, H}[c(u) \neq c(v)] = -x + \frac{x + y}{2} + \frac{1 - y}{\pi} \arccos \rho_{xyz}.
        \]
        By~\cref{lem:bound}, we have $((1 - y) / \pi) \arccos \rho_{xyz} \geq (1 - z + x - y) / 4$, so
        \[
        \Pr_{a, H}[c(u) \neq c(v)] \geq -x + \frac{x + y}{2} + \frac{1 - z + x - y}{4} = \frac{1 - z + y - x}{4}.
        \]
    \item[Case 2: $\bm{x > 0}$.] In this case, with probability $x$ we have $a \leq |x| \leq |y|$, and hence $c(u) = 1, c(v) = 1$. With probability $|y| - |x| = y - x$, we have $|x| < a \leq |y|$, in which case $c(v) = 1$ but $c(u)$ is uniformly randomly distributed in $\{ \pm 1 \}$. With probability $1 - |y| = 1 - y$ we have $a > |y| \geq |x|$, in which case $\vy_u, \vy_v$ are separated by $H$ with probability $(1 / \pi) \arccos \rho_{xyz}$. So, applying~\cref{lem:bound} as before, we have
    \[
        \Pr_{a, H}[c(u) \neq c(v)] = \frac{y - x}{2} + \frac{1 - y}{\pi} \arccos \rho_{xyz} \geq \frac{ y - x}{2} + \frac{1 - z + x - y}{4} \geq
        \frac{1 - z + y - x}{4}.
    \]
    \end{description}
    Hence our conclusion follows in all cases.
\end{proof}

{\small
\bibliographystyle{alphaurl}
\bibliography{nz}
}

\end{document}